\documentclass[journal]{IEEEtran}

\usepackage{epstopdf}
\ifCLASSINFOpdf
   \usepackage[pdftex]{graphicx}
   \graphicspath{{C:../pdf/}{C:../jpeg/}}
   \DeclareGraphicsExtensions{.pdf,.jpeg,.png}
\else
  \usepackage[dvips]{graphicx}
  \graphicspath{{C:../eps/}}
  \DeclareGraphicsExtensions{.eps}
\fi
\usepackage{amsmath}
\makeatletter

\newcommand{\Rmnum}[1]{\expandafter\@slowromancap\romannumeral #1@}
\makeatother
\usepackage{bm}
\usepackage{colortbl}
\usepackage{makecell}
\usepackage{multirow}
\usepackage{amssymb}
\usepackage[colorlinks, linkcolor=black, anchorcolor=black, citecolor=black]{hyperref}
\usepackage{threeparttable}
\usepackage{array}
\usepackage{float}
\usepackage{amsmath}
\newtheorem{theorem}{Theorem}

\newtheorem{proof}{Proof}[section]
\usepackage{color}
\usepackage{cite}
\usepackage{algorithm}
\usepackage{algpseudocode}
\usepackage{graphics}
\usepackage{epsfig}

\begin{document}

\title{Optimal Cooperative Driving at Signal-Free Intersections with Polynomial-Time Complexity}%

\author{Huaxin Pei, Yuxiao Zhang, Yi Zhang, Shuo Feng
\thanks{Manuscript received in December 1, 2020; This work was supported in part by National 135 Key R\&D Program Projects under Grant 2018YFB1600600 and National Natural Science Foundation of China under Grant 61673233. (\emph{Corresponding author: Shuo Feng})}
\thanks{Huaxin Pei is with the Department of Automation, Tsinghua University, Beijing China, 100084 (e-mail: phx17@mails.tsinghua.edu.cn).}
\thanks{Yuxiao Zhang is with the Department of Computer Science and Technology, Tsinghua University, Beijing China, 100084 (e-mail: zhang-yx17@tsinghua.org.cn).}
\thanks{Yi Zhang is with the Department of Automation, BNRist, Tsinghua University, Beijing 100084, and Berkeley Shenzhen Institute (TBSI) Shenzhen 518055, China, and also with Jiangsu Province Collaborative Innovation Center of Modern Urban Traffic Technologies, Nanjing 210096, China (e-mail: zhyi@tsinghua.edu.cn).}
\thanks{Shuo Feng is with the Department of Civil and Environmental Engineering, University of Michigan, Ann Arbor, MI 48109 USA (e-mail: fshuo@umich.edu).}}


\maketitle
\begin{abstract}
  Cooperative driving at signal-free intersections, which aims to improve driving safety and efficiency for connected and automated vehicles, has attracted increasing interest in recent years. However, existing cooperative driving strategies either suffer from computational complexity or cannot guarantee global optimality. To fill this research gap, this paper proposes an optimal and computationally efficient cooperative driving strategy with the polynomial-time complexity. By modeling the conflict relations among the vehicles, the solution space of the cooperative driving problem is completely represented by a  newly designed small-size state space. Then, based on dynamic programming, the globally optimal solution can be searched inside the state space efficiently. It is proved that the proposed strategy can reduce the time complexity of computation from exponential to a small-degree polynomial. Simulation results further demonstrate that the proposed strategy can obtain the globally optimal solution within a limited computation time under various traffic demand settings.

\end{abstract}

\IEEEpeerreviewmaketitle
\section{Introduction}
\IEEEPARstart{W}{ith} the help of vehicle-to-everything (V2X) technologies, connected and automated vehicles (CAVs) can share the driving states with the adjacent vehicles. In such case, cooperative driving for CAVs emerges as a promising way to improve traffic safety and efficiency \cite{FengSpatiotemporal,MengAnalysis,LiCooperative,FengString,yu2019managing,Malikopoulos2019,guler2014using}. It has been pointed out in \cite{13,MengAnalysis} that the key problem of cooperative driving is to determine the optimal sequence of vehicles passing through the conflict areas that mainly include on-ramp and signal-free intersection areas.
\par Cooperative driving at on-ramps has been well discussed in recent studies \cite{Rios-Torres2017-1,15, Group, 32,Jing2019}. However, compared with the merging problem, there exist much more complex conflict relations of vehicles at signal-free intersection, which leads to complicated interactions between vehicles. Thus, cooperative driving at intersections will fail by directly applying the strategies of merging problem, and further investigation is required. Generally, existing studies of planning passing sequence of vehicles at signal-free intersections can be classified into two categories according to the optimality, i.e., optimal and sub-optimal strategies.
\par The strategies aiming to find the globally optimal solution can be regarded as the \emph{\textbf{optimal strategy}}. One typical type of the optimal strategies is to formulate a large-scale mixed-integer programming (MIP) problem to resolve all conflicts between vehicles, where each conflict introduces a binary variable to mathematically describe the relative passing sequence \cite{2,M2016Intersection,Ahn2016Safety}. However, the computational complexity increases exponentially with the increasing number of conflicts between vehicles, which causes the ``curse of dimensionality''. Another type of optimal strategies is adopting a string to represent the vehicle sequence and then describe the complete solution space. For example, Li \emph{et al.} \cite{LiCooperative} formulated a spanning tree and proposed a pruning rule to search the globally optimal solution. However, due to the complex conflict relations between vehicles, it still suffers from the state space explosion and thus is intractable for real-time applications.
\par To improve the computational efficiency, the so-called \emph{\textbf{sub-optimal strategy}} usually stops at a sub-optimal solution within a limited computation time. Autonomous intersection management (AIM) \cite{Dresner2004Multiagent, Dresner2008A} and reservation-based strategy \cite{Huang2012Assessing, Choi2018Reservation} are the typical sub-optimal strategies utilized in the problem of cooperative driving, where heuristic rules are used to instruct the vehicles passing through the intersection roughly in first-in-first-out (FIFO) manner \cite{15}. Nevertheless, as illustrated in \cite{MengAnalysis, 6}, these strategies have a poor performance in improving traffic efficiency. In \cite{xu2018distributed}, the vehicles from different directions are projected into a virtual lane to describe and resolve the conflicts between vehicles. This method can significantly reduce the computational complexity, but it theoretically leads to a sub-optimal solution and thus cannot guarantee traffic efficiency. Recently, several promising strategies are proposed to keep a good balance between computational complexity and traffic efficiency. In \cite{16, 24}, Monte Carlo tree search (MCTS) method combining with heuristic rules is first introduced to the planning problem of vehicle sequence at a signal-free intersection. Numerical simulation results indicate that these strategies can obtain a near-optimal passing sequence within a limited computation time. However, the performance of these strategies lacks a rigorous theoretical guarantee and further analysis is necessary for practical applications.
\par To overcome the limitations of these existing optimal and sub-optimal strategies, we propose an optimal and computationally efficient cooperative driving strategy with polynomial-time complexity in this paper. By modeling the conflict relations among the vehicles, the solution space of the cooperative driving problem is completely represented by a newly designed small-size state space. Then, based on dynamic programming, the globally optimal solution can be searched inside the state space efficiently. It can reduce the time complexity of computation from exponential to a small-degree polynomial.
\par This paper significantly extends our previous work for the on-ramp merging problem \cite{15}, where a dynamic programming method was proposed to obtain the globally optimal merging solution with polynomial-time complexity. However, the method cannot be directly applied to signal-free intersections, because of the following challenges: First, the number of conflicts between vehicles significantly increases, as the interactions between vehicles become more complex. More importantly, both the conflict and conflict-free pairs of vehicles exist in the conflict area of an intersection, which is denoted by the heterogeneous conflict relations in this paper, namely, more than one vehicle may pass the conflict area at the same time without collision. In such case, state transition could not guarantee the Markov property and thus dynamic programming is invalidated \cite{IEEEexample:Baldwin2009Principles,denardo2012dynamic}.
\par In this paper, we rebuild the dynamic programming model for signal-free intersections. To keep Markov property, a novel state transition strategy is designed to model the conflict-free pairs of vehicles, where multiple non-conflict vehicles can be assigned with the right of way instead of dealing with only one vehicle during a state transition. It guarantees that each state transition explicitly represents one kind of conflicts between vehicles so that the objective value of the current state is only determined by its predecessor state, i.e., the state transition satisfies Markov property. Moreover, the number of states and state transitions of the constructed state space is restricted as a small-degree polynomial with the number of vehicles. Therefore, dynamic programming is reactivated at signal-free intersections, and we can obtain the globally optimal solution with the polynomial-time complexity of computation. It is a significant improvement compared with most existing studies about cooperative driving at signal-free intersections (e.g., \cite{LiCooperative,16,24}), where the size of solution space increases exponentially with the increasing number of vehicles. It is also worth emphasizing that the proposed strategy is universal in most traffic scenarios where both the conflict and conflict-free pairs of vehicles exist.
\par Theoretical analysis is proposed to justify the computational time complexity of the proposed optimal strategy. It is proved that the optimal strategy has the polynomial-time complexity of computation, which has a lower bound $O(N^4)$ and an upper bound $O(N^6)$, where $N$ denotes the total number of vehicles within consideration. It is the theoretical foundation for overcoming the limitations of existing strategies. Furthermore, simulation results also demonstrate that the proposed strategy can obtain the globally optimal solution within a limited computation time, comparing with the existing strategies under various traffic demand settings.
\par The main contributions of this paper include: (a) we construct a novel state space with a much smaller size than that in the existing studies to describe the complete solution space of cooperative driving at signal-free intersections; (b) we realize optimal cooperative driving at signal-free intersections, which overcomes the limitations of the existing optimal and sub-optimal strategies regarding global optimality and computational efficiency; (c) we give a rigorous theoretical analysis about the computational complexity of the proposed strategy, which is usually lacking in most existing studies.
\par The rest of this paper is organized as follows. \emph{Section \Rmnum{2}} presents the typical signal-free intersection and formulates the cooperative driving problem. \emph{Section \Rmnum{3}} proposes the optimal strategy for cooperative driving at intersections. Then, \emph{Section \Rmnum{4}} gives theoretical analysis about the computational complexity, and simulation results are in \emph{Section \Rmnum{5}}. Finally, the conclusion and further works are presented in \emph{Section \Rmnum{6}}.
\begin{figure}[t]
\centering
\includegraphics[width=2.5in]{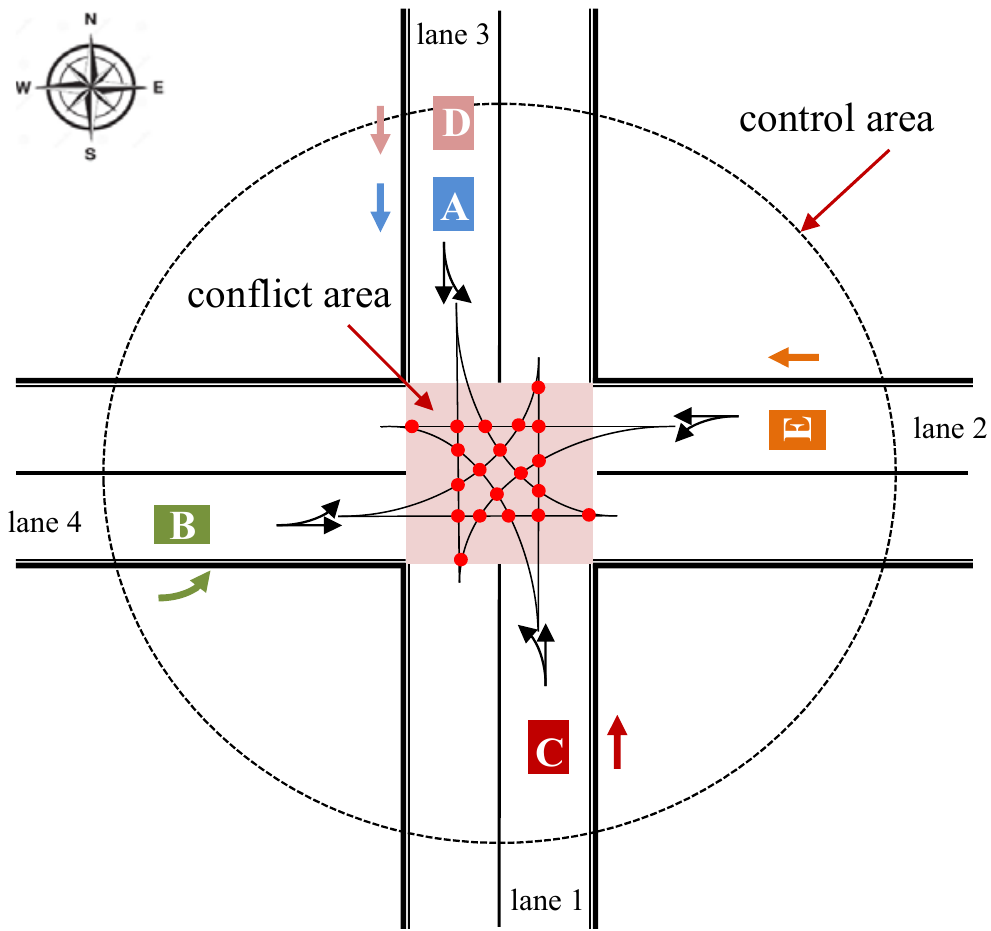}
\caption{A typical signal-free intersection scenario with five vehicles.}
\label{fig1}
\end{figure}
\section{Problem Formulation}
\label{sec2}
\subsection{Scenario and Notations}
\label{sec2-A}
\par In this paper, we select a typical intersection scenario to introduce our method, as shown in Fig. \ref{fig1}. The red area is the conflict area, where the vehicles from different directions may collide. The intersection has a control area, where the vehicles within the control area are considered into the cooperative driving problem. Actually, traffic collisions and delays mainly occur on the vehicles with left-turn and straight movements, and the right-turn vehicles can usually move freely around the intersection. Therefore, to simplify the descriptions, we only consider the left-turn and straight vehicles in this paper, and the scenario including right-turn vehicles can be easily tackled in a similar way.
\par Several reasonable assumptions about the studied scenario are added as follows: a) all vehicles are connected and automated (CAVs); b) lane-change behavior of vehicles within the control area is not allowed. Each vehicle is given a unique identity after arriving at the control area, and $CAV_i$ means the $i^{\text {th }}$ vehicle that reaches the control area. The vehicle identity sets of lane 1, 2, 3 and 4 are
denoted by $\mathbb{N}_1 = \{1,..., \hat n_1\}$, $\mathbb{N}_2 = \{1,..., \hat n_2\}$, $\mathbb{N}_3= \{1,...,\hat n_3\}$ and $\mathbb{N}_4= \{1,...,\hat n_4\}$ respectively, where $\hat n_1$, $\hat n_2$, $\hat n_3$ and $\hat n_4$ denote the total number of vehicles on lane 1, lane 2, lane 3 and lane 4 respectively. The main notations in this paper are shown in TABLE \ref{table1}.
\begin{table*}
\renewcommand{\arraystretch}{1.3}
\centering
\caption{The Nomenclature List}
\centering
\begin{tabular}{c|l}
\Xhline{1.2pt}
{\bfseries Variables} & {\bfseries Notations}\\
\Xhline{1.2pt}
$CAV_i$ & The $i^{th}$ vehicle that reaches the control area.\\

$\hat n_1$, $\hat n_2$, $\hat n_3$, $\hat n_4$ & The total number of vehicles on lane 1, lane 2, lane 3 and lane 4, respectively.\\

$v_{\max}$, $v_{\min}$ & The maximal and minimal velocity of vehicles.\\

$a_{\max}$, $a_{\min}$ & The maximal and minimal acceleration of vehicles.\\

$t_{\text{assign},i}$ & The arrival time assigned to $CAV_i$ to move into the control area.\\

$b_{i,j}$ & The binary decision variable used to formulate the collision avoidance constrains between $CAV_i$ and $CAV_j$.\\

$\Delta_{t,1}$, $\Delta_{t,2}$ & The minimal allowable safe gaps used in constraints (\ref{equ3})-(\ref{equ5}).\\

$l_c$ & The length of the control area.\\

$\mathbb{N}_1$, $\mathbb{N}_2$, $\mathbb{N}_3$, $\mathbb{N}_4$ & The vehicles identity sets of lane 1, lane 2, lane 3 and lane 4, respectively.\\

$\mathbb{T}^{\text{assign}}$ & The set of the arrival time assigned to all vehicles.\\

$\mathbb{B}$ & The binary decision variable set.\\
\Xhline{0.5pt}
\multirow{5}{*}{$s_{r}\left(\begin{array}{ll}n_1 & n_2 \\ n_3 & n_4\end{array}\right)$ }
& $n_1$: The accumulated number of vehicles added to the current candidate sequence on lane 1 up to current state.\\
& $n_2$: The accumulated number of vehicles added to the current candidate sequence on lane 2 up to current state.\\
& $p_3$: The accumulated number of vehicles added to the current candidate sequence on lane 3 up to current state.\\
& $n_4$: The accumulated number of vehicles added to the current candidate sequence on lane 4 up to current state.\\
& $r$: The lane number that has been assigned with the right of way at current state.\\
\Xhline{1.2pt}
\end{tabular}
\label{table1}
\end{table*}
\subsection{Cooperative Driving Problem}
\label{sec2-B}
\par As illustrated in \cite{MengAnalysis,LiCooperative}, cooperative driving aims to improve traffic efficiency and safety by planning the sequence of vehicles passing through the conflict area. To reach this goal, we formulate the objective function as
\begin{subequations}
\begin{equation}
    \min_{\mathbb{T}^{\text {assign}}} J,
\label{equ1a}
\end{equation}
\begin{equation}
J=\max_{i \in \mathbb{N}_1\mathop{\cup}\mathbb{N}_2\mathop{\cup}\mathbb{N}_3\mathop{\cup}\mathbb{N}_4} t_{\text{assign},i},
\label{equ1b}
\end{equation}
\begin{equation}
t_{\text{assign},i} \in \mathbb{T}^{\text {\text{assign}}},
\label{equ1c}
\end{equation}
\end{subequations}
\par\noindent where the decision variable $t_{\text{assign},i}$ denotes the arrival time assigned to $CAV_i$ to move into the conflict area, and $\mathbb{T}^{\text {assign}}$ is the set of the arrival time assigned for all vehicles within the control area. Obviously, $J$ denotes the total passing time of all vehicles to pass through the conflict area \cite{21,15}.
\par The arrival time assigned to each vehicle needs to satisfy vehicle dynamics, i.e.,
\begin{equation}
    t_{\text{assign},i} \geq t_{\min,i},
\label{equ2}
\end{equation}
where $t_{min,i}$ denotes the minimal arrival time assigned to vehicle $i$, which can be easily calculated according to vehicle dynamics \cite{15,16}.
\par As for the vehicles moving on the same lane, the arrival time assigned to each vehicle should satisfy the physical constraints to avoid rear-end collision, i.e.,
\begin{equation}
    t_{\text{assign},i} - t_{\text{assign},j} \geq \Delta_{t,1},
\label{equ3}
\end{equation}
where $\Delta_{t,1}$ denotes the minimal safe gap between two consecutive vehicles to avoid rear-end collision. In addition, vehicle $CAV_j$ is physically ahead of $CAV_i$. For example, $CAV_A$ and $CAV_D$ in Fig. \ref{fig1}.
\par The vehicles from different lanes may have collisions in the conflict area. The method to avoid these collisions is to schedule the vehicles moving into the conflict area sequentially. Thus, the binary variables are introduced to describe the constraints of vehicle sequence, i.e.,
\begin{equation}
    t_{\text{assign},i} - t_{\text{assign},j} + \bold{M} \cdot b_{i,j} \geq \Delta_{t,2},
\label{equ4}
\end{equation}
\begin{equation}
    t_{\text{assign},j} - t_{\text{assign},i} + \bold{M} \cdot (1-b_{i,j}) \geq \Delta_{t,2},
\label{equ5}
\end{equation}
where $\Delta_{t,2}$ denotes the minimal safe gap to avoid collisions between vehicles from different lanes. $CAV_i$ and $CAV_j$ are two vehicles from different lanes and have conflicts, e.g., $CAV_A$ and $CAV_B$ in Fig. \ref{fig1}. $\bold{M}$ is a positive and sufficiently large number. $b_{i,j}$ is the binary variable and we use $\mathbb{B}$ to denote the set of all binary variables. Obviously, $b_{i,j}=1$ implies that vehicle $i$ enters conflict area earlier than vehicle $j$.
\par Based on the above descriptions, the planning problem of vehicle sequence can be mathematically formulated as a mixed-integer programming problem (MIP), as shown in (\ref{equ6}).
\begin{equation}\begin{array}{c}
\min_{\mathbb{T}^{\text {assign}}, \mathbb{B}} J \\
\emph{s.t.}(2)(3)(4)(5)
\label{equ6}
\end{array}\end{equation}
\subsection{Problem Challenges}
\par Compared with the on-ramp merging problem that has been widely investigated, there are two major challenges of solving the cooperative driving problem at signal-free intersections:
\subsubsection{Increased number of conflicts} the number of conflicts between vehicles has increased a lot. As shown in Fig. \ref{fig1}, each red dot represents one kind of conflict pairs, and each conflict between vehicles will introduce a binary variable in the problem (\ref{equ6}). Thus, the size of solution space increases exponentially with the increasing number of conflicts \cite{2,Chen2016,Beale1979Branch,Rios-Torres2017}.
\begin{figure}[t]
\centering
\includegraphics[width=3.3in]{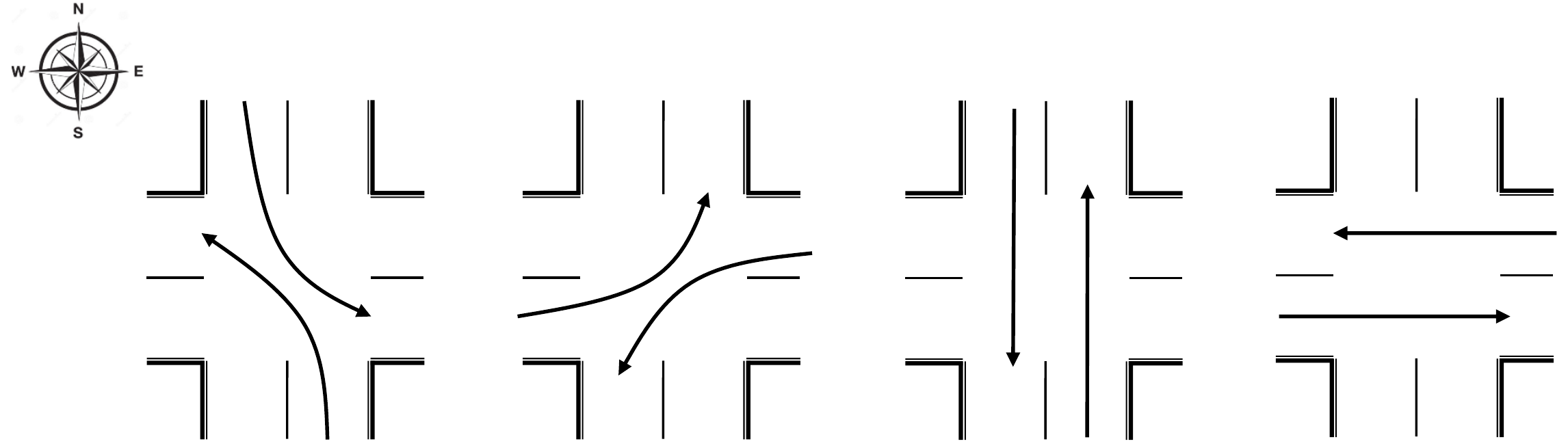}
\caption{Four kinds of collision-free pairs of vehicles at an intersection.}
\label{fig2}
\end{figure}
\subsubsection{Heterogeneous conflict relations} besides the conflict pairs of vehicles, there also exist several collision-free pairs of vehicles shown in Fig. \ref{fig2}, where the corresponding vehicles are allowed to move into the conflict area at the same time. The heterogeneous conflict relations between vehicles further increase the number of feasible solutions and lead to state-space explosion for the state-based methods, e.g., \cite{LiCooperative}.
\par To this end, it is necessary to construct a new cooperative driving strategy to overcome the above challenges and then guarantee computational efficiency and global optimality.
\section{Optimal Cooperative Driving Design}
\label{sec3}
\par In this section, we will reformulate problem (\ref{equ6}) based on dynamic programming to implement optimal cooperative driving with polynomial computational complexity. The key idea is to construct a small-size state space to describe the complete solution space of the large-scale planning problem (\ref{equ6}) and then search for the globally optimal solution by dynamic programming. An intersection scenario with five vehicles shown in Fig. \ref{fig3-new}(a) is taken as an example to introduce our method.
\par The rest of \emph{Section \Rmnum{3}} is organized as follows. The state space describing the complete solution space is constructed in \emph{Section \Rmnum{3}-A}, \emph{Section \Rmnum{3}-B} and \emph{Section \Rmnum{3}-C}. \emph{Section \Rmnum{3}-D} introduces the method to search the globally optimal solution inside the constructed state space. In \emph{Section \Rmnum{3}-E}, state space construction process and solution searching method are integrated into one overall optimal cooperative driving algorithm to save the computational resources.
\subsection{State Definition}
\par The problem of planning vehicle sequence passing through the conflict area is equivalent to sequentially assign the right of way for each vehicle to move into the conflict area. Thus, the problem (\ref{equ6}) can be reformulated as a multi-stage decision process, where the vehicles are sequentially added to the current candidate sequence and only one vehicle is tackled in each stage. Then, as shown in Fig. \ref{fig3-new}(b), state variable is defined to describe the assignment of the right of way, i.e.,
\begin{equation}
    s_{r}\left(\begin{array}{ll}n_1 & n_2 \\ n_3 & n_4\end{array}\right),
\label{equ7}
\end{equation}
where $n_1$, $n_2$, $n_3$ and $n_4$ denote the accumulated number of vehicles added to current candidate sequence in lane 1, 2, 3 and 4 up to current state, respectively. $r$ denotes the lane identity of the vehicle which obtains the right of way at current state. For instance, $s_{3}\begin{pmatrix}\begin{smallmatrix}1 & 0 \\ 2 & 0\end{smallmatrix}\end{pmatrix}$ denotes that one vehicle in lane 1 and two vehicles in lane 3 have been tackled up to current state. In addition, the vehicle with the right of way at current state can be exactly specified. For instance, $s_{3}\begin{pmatrix}\begin{smallmatrix}1 & 0 \\ 2 & 0\end{smallmatrix}\end{pmatrix}$ denotes that the second vehicle in lane 3 ($r=3$, $n_r=2$) has the right of way at the current state, i.e., $CAV_D$.
\subsection{Basic State Transition}
\par State transition is used to describe the transition of the right of way between vehicles. We use the lane identity as the decision variable $u$ $(u \in \{1,2,3,4\})$ to determine which lane has obtained the right of way at the current stage so that one vehicle on the corresponding lane is added to the candidate sequence. The state transition function emerges after introducing the state variable and decision variable, i.e.,
\begin{equation}
    s_{r}\begin{pmatrix}\begin{smallmatrix}n_1 & n_2 \\ n_3 & n_4\end{smallmatrix}\end{pmatrix}=g\left(s_{r'}\begin{pmatrix}\begin{smallmatrix}n_1' & n_2' \\ n_3' & n_4'\end{smallmatrix}\end{pmatrix},u\right),
\label{equ8}
\end{equation}
where $s_{r'}\begin{pmatrix}\begin{smallmatrix}n_1' & n_2' \\ n_3' & n_4'\end{smallmatrix}\end{pmatrix}$ is the predecessor state of $s_{r}\begin{pmatrix}\begin{smallmatrix}n_1 & n_2 \\ n_3 & n_4\end{smallmatrix}\end{pmatrix}$. The function $g(\cdot)$ denotes the state transition function, i.e., if $u=i$, we have $n_i = n_i'+1$, $r = i$ and other parameters of the state variable stay the same. For example, as shown in Fig. \ref{fig3-new}(c), the initial state $s_{0}\begin{pmatrix}\begin{smallmatrix}0 & 0 \\ 0 & 0\end{smallmatrix}\end{pmatrix}$ connects to different successor states according to different values of the decision variable.
\par A new vehicle will be assigned an arrival time to the conflict area if it obtains the right of way during a state transition. Thus, the objective value $J$ varies with state transition. As for the conflict pairs of vehicles, the arrival time assigned to the vehicle specified by current state is always larger than that assigned to the vehicle specified by its predecessor state, because the vehicle specified by the predecessor state has the priority over the vehicle specified by current state to enter the conflict area during each state transition. Therefore, the objective value $J$ up to current state is exactly equal to the time assigned to the vehicle specified by current state, i.e.,
\begin{equation}
J(s)=t_{\text{assign},s},
\label{equ9'}
\end{equation}
where $J(s)$ denotes the objective value $J$ up to state $s$, and $t_{\text{assign},s}$ denotes the arrival time assigned to the vehicle specified by state $s$. Then, considering constraints (\ref{equ2})-(\ref{equ5}), the recurrence function of objective value $J$ during state transition can be summarized as
\begin{equation}
\label{equ9}
\begin{aligned}
    J(s)&=t_{\text{assign},s}\\&=\max(t_{\min,s},t_{\text{assign},s'}+\Delta_T)\\&=\max(t_{\min,s},J\left(s'\right)+\Delta_T),
\end{aligned}
\end{equation}
where $t_{\min,s}$ denotes the minimal arrival time constrained by vehicle dynamics of the vehicle specified by state $s$, and $s'$ is the predecessor state of $s$. If the vehicles specified by $s$ and $s'$ are in the same lane, $\Delta_T=\Delta_{t,1}$, otherwise, $\Delta_T=\Delta_{t,2}$.
\par According to Eq. (\ref{equ9}), we can find that the objective value $J(s)$ is directly determined by its predecessor state $s'$ and has nothing with other states before the predecessor state. In other words, state transition satisfies Markov property when there exist conflicts between the vehicles specified by the corresponding two states, which is a compulsory requirement of using dynamic programming to solve a multi-stage decision problem \cite{IEEEexample:Baldwin2009Principles, denardo2012dynamic}.
\par However, as for the conflict-free pairs of vehicles presented in \emph{Section \Rmnum{2}-C}, the objective value $J$ up to current state has nothing with that of its predecessor state, because they are not conflict with each other. As shown in Fig. \ref{fig3-new}(d),  $CAV_C$ specified by $s_{1}\begin{pmatrix}\begin{smallmatrix}1 & 1 \\ 2 & 0\end{smallmatrix}\end{pmatrix}$ and $CAV_D$ specified by $s_{3}\begin{pmatrix}\begin{smallmatrix}0 & 1 \\ 2 & 0\end{smallmatrix}\end{pmatrix}$ are the so called conflict-free pairs so that the arrival time assigned to $CAV_C$ is determined by that assigned to the vehicles other than $CAV_D$. In other words, $J\left(s_{1}\begin{pmatrix}\begin{smallmatrix}1 & 1 \\ 2 & 0\end{smallmatrix}\end{pmatrix}\right)$ is not determined by its predecessor state $s_{3}\begin{pmatrix}\begin{smallmatrix}0 & 1 \\ 2 & 0\end{smallmatrix}\end{pmatrix}$, but by the states before $s_{3}\begin{pmatrix}\begin{smallmatrix}0 & 1 \\ 2 & 0\end{smallmatrix}\end{pmatrix}$. Therefore, the state transition cannot satisfy Markov property and thus invalidates dynamic programming in this problem. To this end, it is necessary to reconstruct state transition function to overcome the challenges from heterogeneous conflict relations.
\begin{figure*}
\centering
\includegraphics[width=7.2in]{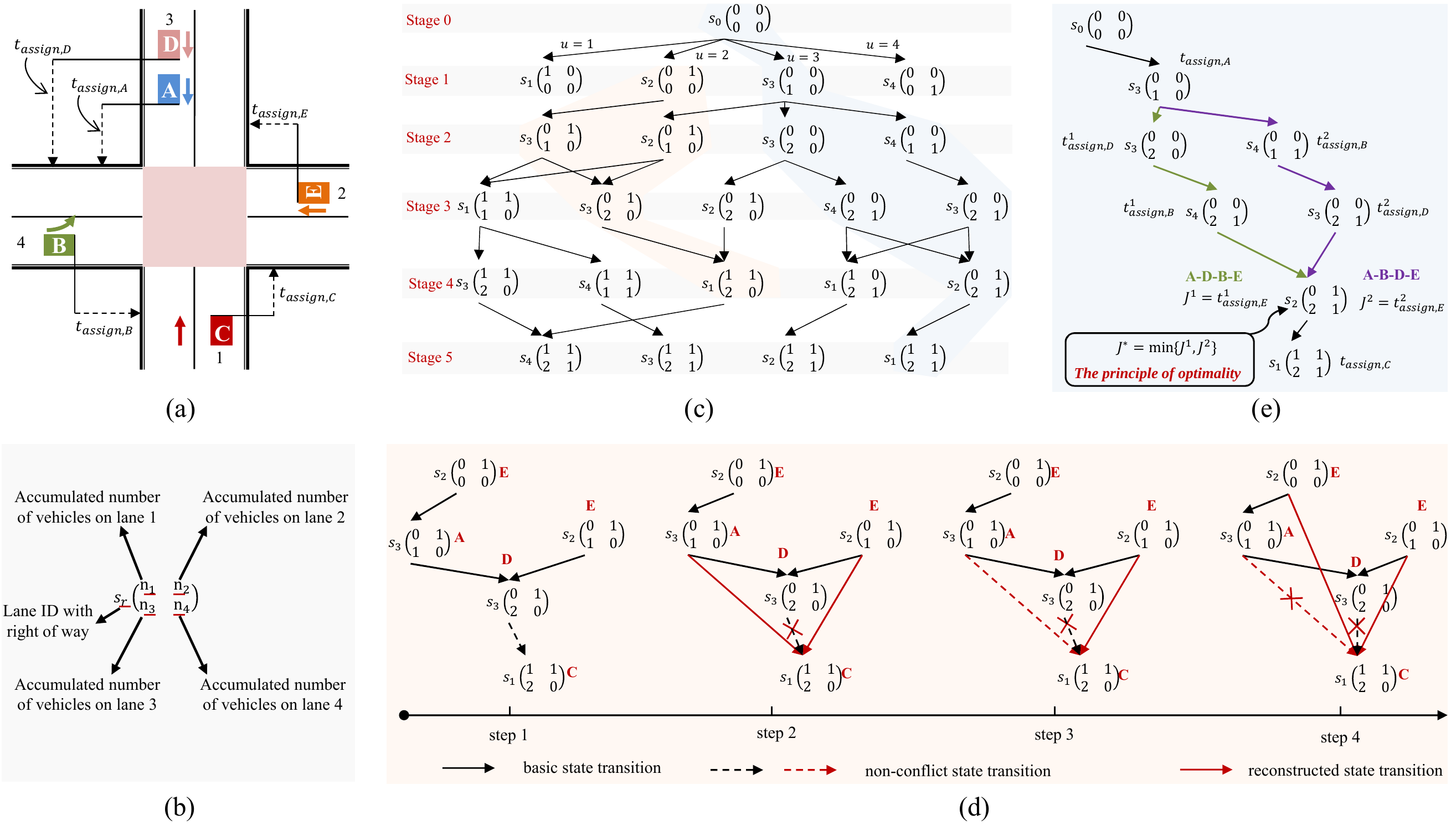}
\caption{Optimal cooperative driving strategy applied to a simple signal-free intersection scenario with five vehicles. Fig. 3(a) presents the studied scenario. Fig. 3(b) gives the definition of state variable. Fig. 3(c) presents a part of state space under basic state transition, and many states and state transitions are omitted for simplicity. Fig. 3(d) illustrates how to reconstruct the state transitions. Fig. 3(e) is the decision-making process.}
\label{fig3-new}
\end{figure*}
\subsection{State Transition Reconstruction}
\par In this subsection, we will reconstruct the state transition function to keep Markov property then reactivate dynamic programming for cooperative driving at signal-free intersections.
\par To keep Markov property, the non-conflict transitions (i.e., the transitions which connect the states that are not in conflict with each other) should be removed from the state space. To guarantee global optimality, we need to build new transitions for the corresponding states to completely describe the assignment of the right of way. Firstly, we need to find the non-conflict transitions in state space and then, for each non-conflict state transition, we need to remove it and connect current state with all predecessor states of its original predecessor state. Repeat the above steps until there is no non-conflict transition in state space. The reconstruction strategy guarantees that each state transition represents one kind of conflict between vehicles and all possible assignments of the right of way are presented in state space.
\par Here we illustrate our key idea in Fig. \ref{fig3-new}(d) and will elaborate the implementation algorithm in \emph{Section \Rmnum{3}-E}. Firstly, we find that the transition between state $s_{3}\begin{pmatrix}\begin{smallmatrix}0 & 1 \\ 2 & 0\end{smallmatrix}\end{pmatrix}$ and $s_{1}\begin{pmatrix}\begin{smallmatrix}1 & 1 \\ 2 & 0\end{smallmatrix}\end{pmatrix}$ is the so called non-conflict transition as $CAV_D$ and $CAV_C$ can enter the conflict area at the same time, and thus we need to remove this transition and establish new transitions for $s_{1}\begin{pmatrix}\begin{smallmatrix}1 & 1 \\ 2 & 0\end{smallmatrix}\end{pmatrix}$. Secondly, the non-conflict transition between $s_{3}\begin{pmatrix}\begin{smallmatrix}0 & 1 \\ 2 & 0\end{smallmatrix}\end{pmatrix}$ and $s_{1}\begin{pmatrix}\begin{smallmatrix}1 & 1 \\ 2 & 0\end{smallmatrix}\end{pmatrix}$ is removed, and $s_{1}\begin{pmatrix}\begin{smallmatrix}1 & 1 \\ 2 & 0\end{smallmatrix}\end{pmatrix}$ connects to the predecessor states of $s_{3}\begin{pmatrix}\begin{smallmatrix}0 & 1 \\ 2 & 0\end{smallmatrix}\end{pmatrix}$, i.e., $s_{3}\begin{pmatrix}\begin{smallmatrix}0 & 1 \\ 1 & 0\end{smallmatrix}\end{pmatrix}$ and $s_{2}\begin{pmatrix}\begin{smallmatrix}0 & 1 \\ 1 & 0\end{smallmatrix}\end{pmatrix}$. Thirdly, we find that one of the reproduced transitions is still the non-conflict transition, i.e., the transition between $s_{3}\begin{pmatrix}\begin{smallmatrix}0 & 1 \\ 1 & 0\end{smallmatrix}\end{pmatrix}$ and $s_{1}\begin{pmatrix}\begin{smallmatrix}1 & 1 \\ 2 & 0\end{smallmatrix}\end{pmatrix}$. Finally, the non-conflict transition between $s_{3}\begin{pmatrix}\begin{smallmatrix}0 & 1 \\ 1 & 0\end{smallmatrix}\end{pmatrix}$ and $s_{1}\begin{pmatrix}\begin{smallmatrix}1 & 1 \\ 2 & 0\end{smallmatrix}\end{pmatrix}$ is removed, and $s_{1}\begin{pmatrix}\begin{smallmatrix}1 & 1 \\ 2 & 0\end{smallmatrix}\end{pmatrix}$ connects to the predecessor state of $s_{3}\begin{pmatrix}\begin{smallmatrix}0 & 1 \\ 1 & 0\end{smallmatrix}\end{pmatrix}$, i.e., $s_{2}\begin{pmatrix}\begin{smallmatrix}0 & 1 \\ 0 & 0\end{smallmatrix}\end{pmatrix}$. Obviously, the non-conflict transition between $s_{3}\begin{pmatrix}\begin{smallmatrix}0 & 1 \\ 2 & 0\end{smallmatrix}\end{pmatrix}$ and $s_{1}\begin{pmatrix}\begin{smallmatrix}1 & 1 \\ 2 & 0\end{smallmatrix}\end{pmatrix}$ is modified to guarantee Markov property and global optimality. Compared with the basic state transition in \emph{Section \Rmnum{3}-B}, the reconstruction strategy have following new properties.
\par\emph{Property 1: The state transition is not just between the adjacent stages. Thus, one state transition does not necessarily represent only one vehicle that obtains the right of way while it may represent a group of vehicles in which there is no conflict of different directions. As shown in Fig. \ref{fig3-new}(d) step 4, the transition between $s_{2}\begin{pmatrix}\begin{smallmatrix}0 & 1 \\ 0 & 0\end{smallmatrix}\end{pmatrix}$ and $s_{1}\begin{pmatrix}\begin{smallmatrix}1 & 1 \\ 2 & 0\end{smallmatrix}\end{pmatrix}$ represents that $CAV_A$, $CAV_D$ and $CAV_C$ are assigned with the right of way during current state transition.}
\par\emph{Property 2: Removing one state transition may reproduce multiple transitions. It increases the number of transitions in the reconstructed state space, and theoretical analysis about this change will be presented in \emph{Section \Rmnum{4}}.}

\par By Property 1, the recurrence function of objective value $J$ needs to be modified based on the reconstructed transitions. The objective value $J$ up to current state should consider the assigned arrival time of a group of vehicles specified by the corresponding transition rather than only one vehicle. We assume that state $s'$ is the predecessor state of $s$ and use $\mathbb{G}$ to denote the set of vehicle identity of a group of vehicles that are assigned the right of way during the transition from $s'$ to $s$. For each vehicle in $\mathbb{G}$, the method to assign arrival time is
\begin{subequations}
\begin{equation}
\begin{aligned}
    t_{\text{assign},k} =\max(t_{\min,k},J (s')+\Delta_T,t_{\text{assign},k'}+\Delta_{t,1}),
\label{equ11a}
\end{aligned}
\end{equation}
\begin{equation}
k,k' \in \mathbb{G},
\label{equ11b}
\end{equation}
\end{subequations}
where $CAV_{k'}$ is physically ahead of $CAV_{k}$ in the same lane. If the vehicles specified by $s$ and $s'$ are in the same lane, $\Delta_T=\Delta_{t,1}$, otherwise, $\Delta_T=\Delta_{t,2}$. Then, we can obtain the objective value $J$ up to current state, i.e.,
\begin{equation}
J(s) = \max_{k \in \mathbb{G}} t_{\text{assign},k}.
\label{equ12'}
\end{equation}
Based on (\ref{equ11a}), (\ref{equ11b}) and (\ref{equ12'}), we can find that $J(s)$ is determined by its predecessor state $s'$ and has nothing with other states before state $s'$, i.e., the reconstructed state transition guarantee Markov property. For example, as shown in Fig. \ref{fig3-new}(d), the state transition between $s_{2}\begin{pmatrix}\begin{smallmatrix}0 & 1\\ 0 & 0\end{smallmatrix}\end{pmatrix}$ and state $s_{1}\begin{pmatrix}\begin{smallmatrix}1 & 1 \\ 2 & 0\end{smallmatrix}\end{pmatrix}$ represents that $CAV_A$, $CAV_D$ and $CAV_C$ are assigned with the arrival time at this time, i.e.,
\begin{equation}
\label{equ13'}
t_{\text{assign},A} = \max\left(t_{\min,A},J\left(s_{2}\begin{pmatrix}\begin{smallmatrix}0 & 1\\ 0 & 0\end{smallmatrix}\end{pmatrix}\right)+\Delta_{t,2}\right),
\end{equation}
\begin{equation}
\label{equ14'}
t_{\text{assign},D} = \max(t_{\min,D},t_{\text{assign},A}+\Delta_{t,1}),
\end{equation}
\begin{equation}
\label{equ15'}
t_{\text{assign},C} = \max\left(t_{\min,C},J\left(s_{2}\begin{pmatrix}\begin{smallmatrix}0 & 1\\ 0 & 0\end{smallmatrix}\end{pmatrix}\right)+\Delta_{t,2}\right),
\end{equation}
\begin{equation}
\label{equ16'}
J\left(s_{1}\begin{pmatrix}\begin{smallmatrix}1 & 1 \\ 2 & 0\end{smallmatrix}\end{pmatrix}\right) = \max(t_{\text{assign},A},t_{\text{assign},D},t_{\text{assign},C}).
\end{equation}
According to (\ref{equ13'})-(\ref{equ16'}), it can be found that $J\left(s_{1}\begin{pmatrix}\begin{smallmatrix}1 & 1 \\ 2 & 0\end{smallmatrix}\end{pmatrix}\right)$ is determined by $J\left(s_{2}\begin{pmatrix}\begin{smallmatrix}0 & 1 \\ 0 & 0\end{smallmatrix}\end{pmatrix}\right)$ and has nothing with other states before $s_{2}\begin{pmatrix}\begin{smallmatrix}0 & 1 \\ 0 & 0\end{smallmatrix}\end{pmatrix}$.
\par Based on the aforementioned descriptions, the constructed state space has following properties.
\par\emph{Property 3: There are many phenomena that different vehicle sequences reach the same state in the state space by determining the accumulated assigned number of vehicles in the specific lane, which leads to that the size of the state space is much smaller compared with those methods in which the vehicle sequence is directly used as the state \cite{LiCooperative,16}.}
\par\emph{Property 4: By determining the accumulated assigned number of vehicles in the specific lane, vehicles in the same lane always comply with the physical constraints (\ref{equ3}) during each state transition, which leads to that the infeasible sequences that violate the physical constraints are directly eliminated and the state space contains only all feasible solutions.}
\par Property 3 and Property 4 make it possible to describe the complete solution space of the large-scale planning problem (\ref{equ6}) utilizing a small-size state space. Theoretical analysis about the size of the state space will be presented in \emph{Section \Rmnum{4}}.
\subsection{Decision-Making Process}
\par In this subsection, we will introduce the method to search the globally optimal solution in the constructed state space.
\par As described in \emph{Section \Rmnum{3}-C}, state transition satisfy Markov property in the state space. Thus, \emph{the principle of optimality} can be adopted to search the globally optimal solution.
\par\noindent\emph{Lemma 1(the principle of optimality): if the optimal path of states passes through a particular state $S$ in the state space, the partial path from the initial state to state $S$ must also be the optimal path from the initial state to state $S$ \cite{IEEEexample:Baldwin2009Principles}}.
\par According to Eq. (\ref{equ11a})-(\ref{equ12'}) and Fig. \ref{fig3-new}(e), the objective value is updated and the corresponding vehicles can get the assigned arrival time during each state transition. If there is only one predecessor state of current state, the objective value of current state is determined by that of the unique predecessor state. However, if there are multiple predecessor states, we need to find the optimal objective value of current state according to Lemma 1, i.e.,
\begin{equation}
  J^* = \min_{i \in \mathbb{S}'} J^i,
  \label{equ17}
\end{equation}
 where $\mathbb{S}'$ denotes the set of all predecessor states of current state. $J^i$ denotes the candidate objective value of current state determined by its $i^{th}$ predecessor state, and $J^*$ is the optimal objective value of current state. According to Eq. (\ref{equ17}),  we can get the optimal objective value and record the optimal predecessor state of each state.
\par Based on the above descriptions, it can be easily concluded that the minimal objective value among all terminal states in the last stage is the globally optimal objective value of the problem (\ref{equ6}). Then, the optimal sequence can be obtained via a backtracking search from the optimal terminal state according to the recorded optimal predecessor state of each state.
\renewcommand{\algorithmicrequire}{\textbf{Input:}} 
\renewcommand{\algorithmicensure}{\textbf{Output:}} 
\begin{algorithm}[t]
\caption{Optimal Cooperative Driving Algorithm}
\begin{algorithmic}[1]
\Require
The driving states of all vehicles within control area.
\Ensure
The optimal arrival time assigned to each vehicle.
\For{each $i\in [1,N]$}
\For{each state $s_{1}\begin{pmatrix}\begin{smallmatrix}n_1 & n_2 \\ n_3 & n_4\end{smallmatrix}\end{pmatrix}$ in stage $i$}
\State $N_1=1$, $N_3=0$;
\State $j=n_1$;
\While {$V_{1,j-1}$ has the same intention as $V_{1,n_1}$}
\State $j=j-1$, $N_1=N_1+1$;
\EndWhile
\State $j=n_3$;
\While {$V_{3,j}$ has the same intention as $V_{1,n_1}$}
\State $j=j-1$, $N_3=N_3+1$;
\EndWhile
\State Connect $s_{1}\begin{pmatrix}\begin{smallmatrix}n_1 & n_2 \\ n_3 & n_4\end{smallmatrix}\end{pmatrix}$ with $s_{1}\begin{pmatrix}\begin{smallmatrix}n_1-1 & n_2 \\ n_3 & n_4\end{smallmatrix}\end{pmatrix}$;
\For {each $\overline n_3 \in [1,N_3]$}
\State Connect $s_{1}\begin{pmatrix}\begin{smallmatrix}n_1 & n_2 \\ n_3 & n_4\end{smallmatrix}\end{pmatrix}$ with $s_{1}\begin{pmatrix}\begin{smallmatrix}n_1-N_1 & n_2 \\ n_3-\overline n_3 & n_4\end{smallmatrix}\end{pmatrix}$;
\EndFor
\For {each $\overline n_1 \in [1,N_1]$}
\State Connect $s_{1}\begin{pmatrix}\begin{smallmatrix}n_1 & n_2 \\ n_3 & n_4\end{smallmatrix}\end{pmatrix}$ with $s_{3}\begin{pmatrix}\begin{smallmatrix}n_1-\overline n_1 & n_2 \\ n_3-N_3 & n_4\end{smallmatrix}\end{pmatrix}$;
\EndFor
\For{each $\overline n_1 \in [1,N_1]$ and each $\overline n_3 \in [0,N_3]$}
\State Connect $s_{1}\begin{pmatrix}\begin{smallmatrix}n_1 & n_2 \\ n_3 & n_4\end{smallmatrix}\end{pmatrix}$ with $s_{2}\begin{pmatrix}\begin{smallmatrix}n_1-\overline n_1 & n_2 \\ n_3-\overline n_3 & n_4\end{smallmatrix}\end{pmatrix}$;
\State Connect $s_{1}\begin{pmatrix}\begin{smallmatrix}n_1 & n_2 \\ n_3 & n_4\end{smallmatrix}\end{pmatrix}$ with $s_{4}\begin{pmatrix}\begin{smallmatrix}n_1-\overline n_1 & n_2 \\ n_3-\overline n_3 & n_4\end{smallmatrix}\end{pmatrix}$;
\EndFor
\State Update the objective value for each transition;
\State Record the optimal predecessor state;
\EndFor
\State Other states can be tackled in similar way when $r \neq 1$;
\EndFor
\State Obtain the optimal passing sequence via backtracking search and assign the optimal arrival time for each vehicle.
\end{algorithmic}
\end{algorithm}
\subsection{Algorithm for Implementation}
\par In the above subsections, we introduce our method in three steps for easier understanding, i.e., basic state transition, state transition reconstruction and decision-making process.  As commonly used in dynamic programming, these three steps can be integrated into one overall algorithm when implementing optimal cooperative driving to save computational resources. Specifically, instead of constructing the state space under basic state transition and then reconstructing the non-conflict transitions, we can directly construct the feasible state transitions by combining the basic state transition function and reconstruction strategy. At the same time, the currently optimal objective value of each state can be updated during each state transition. It avoids repeated computation, and the integrated algorithm is shown in Algorithm 1, where $V_{i,j}$ denotes the $j^{th}$ vehicle in lane $i$. $N$ denotes the total number of vehicles within the control area. The driving intention refers to whether the vehicle is going straight or turning left.
\par In Algorithm 1, for each state $s_{r}\begin{pmatrix}\begin{smallmatrix}n_1 & n_2 \\n_3 & n_4\end{smallmatrix}\end{pmatrix}$ in the state space, we need to find all possible predecessor states and then obtain the optimal $J\left(s_{r}\begin{pmatrix}\begin{smallmatrix}n_1 & n_2 \\n_3 & n_4\end{smallmatrix}\end{pmatrix}\right)$ and the optimal predecessor state of $s_{r}\begin{pmatrix}\begin{smallmatrix}n_1 & n_2 \\n_3 & n_4\end{smallmatrix}\end{pmatrix}$. Firstly, we need to determine the vehicles that may obtain the right of way during the state transitions connecting to state $s_{r}\begin{pmatrix}\begin{smallmatrix}n_1 & n_2 \\n_3 & n_4\end{smallmatrix}\end{pmatrix}$, i.e., a continuous queue of vehicles in lane $r$ with the same driving intention as $V_{r,n_r}$ starting from $V_{r,n_r}$ and a continuous queue of vehicles in lane $r_o$ with the same driving intention as $V_{r,n_r}$ starting from $V_{r_o,n_{r_o}}$, where $r_o$ denotes the lane identity of the opposite lane of lane $r$. There is no conflict of different directions among these vehicles and thus it is possible for them to get the right of way during the same state transition. Then, all the predecessor states of $s_{r}\begin{pmatrix}\begin{smallmatrix}n_1 & n_2 \\n_3 & n_4\end{smallmatrix}\end{pmatrix}$ can be divided into five categories according to the lane identity of the vehicle specified by each predecessor state. At the same time, the currently optimal $J\left(s_{r}\begin{pmatrix}\begin{smallmatrix}n_1 & n_2 \\n_3 & n_4\end{smallmatrix}\end{pmatrix}\right)$ is updated during each state transition and we can obtain the optimal predecessor state of $s_{r}\begin{pmatrix}\begin{smallmatrix}n_1 & n_2 \\n_3 & n_4\end{smallmatrix}\end{pmatrix}$. Finally, all the states in the state space can be tackled in the similar way and thus the optimal passing sequence can be exported through backtracking search.
\par Taking Fig. \ref{fig4-new} as the example to illustrate how to find all possible predecessor states of state $s_{1}\begin{pmatrix}\begin{smallmatrix}n_1 & n_2 \\ n_3 & n_4\end{smallmatrix}\end{pmatrix}$. Obviously, there are at most $N_1$ vehicles in lane 1 and $N_3$ vehicles in lane 3 that may obtain the right of way during the state transitions connecting to $s_{1}\begin{pmatrix}\begin{smallmatrix}n_1 & n_2 \\ n_3 & n_4\end{smallmatrix}\end{pmatrix}$. Then, there are five kinds of predecessor states of $s_{1}\begin{pmatrix}\begin{smallmatrix}n_1 & n_2 \\ n_3 & n_4\end{smallmatrix}\end{pmatrix}$: \textcircled{1} if the vehicle specified by the predecessor state is in lane 1 and only one vehicle in lane 1 can obtain the right of way, the predecessor state is $s_{1}\begin{pmatrix}\begin{smallmatrix}n_1-1 & n_2 \\ n_3 & n_4\end{smallmatrix}\end{pmatrix}$. \textcircled{2} if the vehicle specified by the predecessor state is in lane 1 and both the vehicles in lane 1 and lane 3 can obtain the right of way, the predecessor states can be described as $s_{1}\begin{pmatrix}\begin{smallmatrix}n_1-N_1 & n_2 \\ n_3-\overline n_3 & n_4\end{smallmatrix}\end{pmatrix}$, where $\overline n_3 \in [1,N_3]$. \textcircled{3} if the vehicle specified by the predecessor state is in lane 3, the predecessor states can be described as $s_{3}\begin{pmatrix}\begin{smallmatrix}n_1-\overline n_1 & n_2 \\ n_3-N_3 & n_4\end{smallmatrix}\end{pmatrix}$, where $\overline n_1 \in [1,N_1]$. \textcircled{4} if the vehicle specified by the predecessor state is in lane 2, the predecessor states can be described as $s_{2}\begin{pmatrix}\begin{smallmatrix}n_1-\overline n_1 & n_2 \\ n_3-\overline n_3 & n_4\end{smallmatrix}\end{pmatrix}$, where $\overline n_1 \in [1,N_1]$ and $\overline n_3 \in [0,N_3]$. \textcircled{5} if the vehicle specified by the predecessor state is in lane 4, the predecessor states can be described as $s_{4}\begin{pmatrix}\begin{smallmatrix}n_1-\overline n_1 & n_2 \\ n_3-\overline n_3 & n_4\end{smallmatrix}\end{pmatrix}$, where $\overline n_1 \in [1,N_1]$ and $\overline n_3 \in [0,N_3]$.
\begin{figure}[t]
\centering
\includegraphics[width=3.5in]{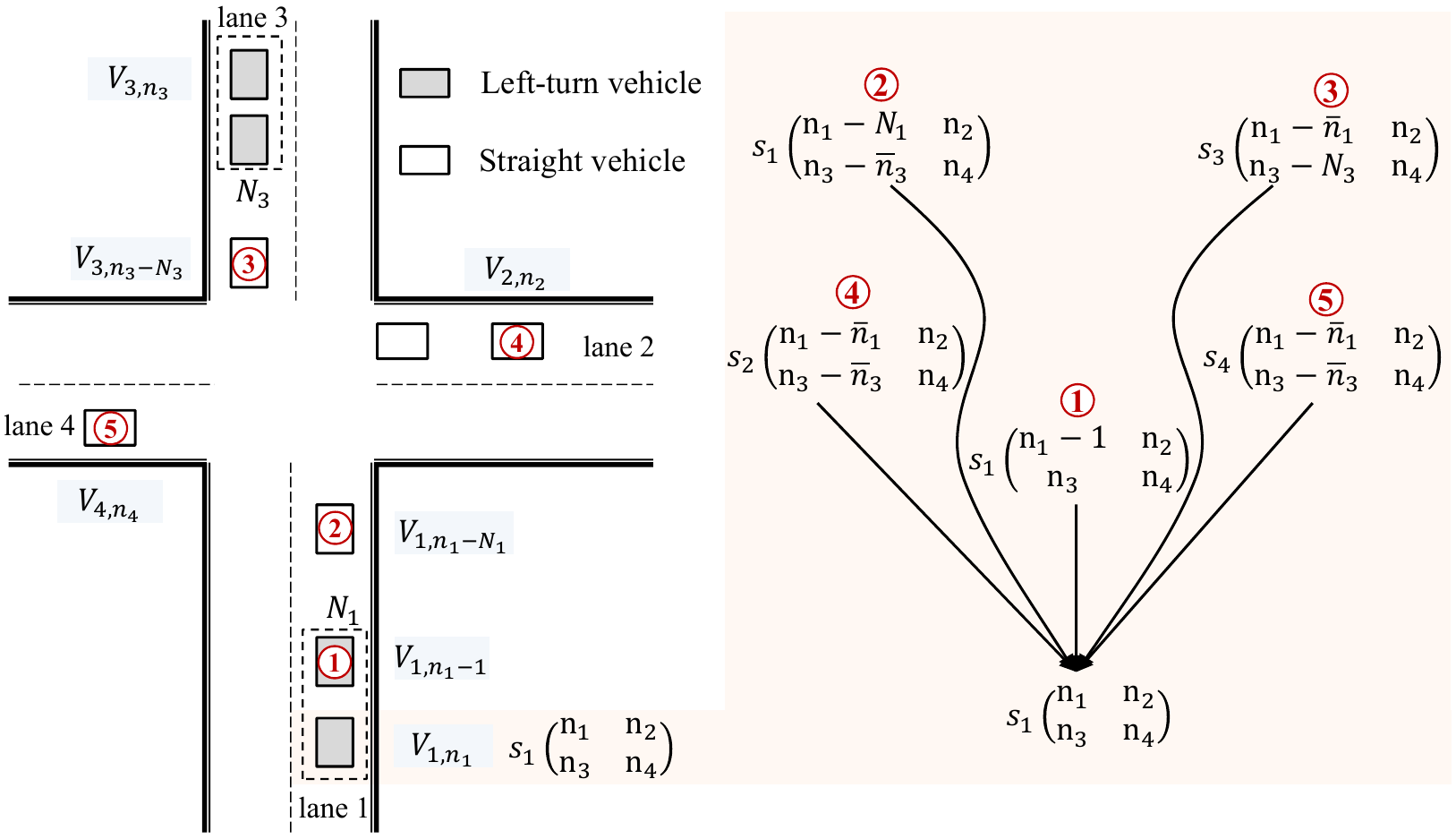}
\caption{An illustration of how to find all predecessor states for each state in Algorithm 1. There are five kinds of predecessor states.}
\label{fig4-new}
\end{figure}
\section{Analysis of Computational Complexity}
\label{sec5}
\par In this section, both the size of the constructed state space and the computational complexity of the proposed strategy are proved through theoretical analysis.
\subsection{Analysis of Number of States}
\label{sec5-A}
\par The theorem is proposed for the number of states.
\begin{theorem}
\emph{The total number of states is $O(N^4)$, where $N$ denotes the total number of vehicles within control area.}
\end{theorem}
\begin{proof}
    See \emph{Appendix \ref{appendix A}}.
\end{proof}
\par By Theorem 1, the number of states is a quartic polynomial with the number of vehicles. Thus, the size of state space grows slowly with the increasing number of vehicles.
\subsection{Analysis of Number of Transitions}
\label{sec5-B}
\par By Property 2 in \emph{Section \Rmnum{3}-C}, it can be concluded that the number of transitions varies with the conflict relations between vehicles within the conflict area. Thus, there exist the lower and upper bounds on the number of state transitions: if all the vehicles within the control area are in conflict with each other, the number of transitions reaches the lower bound, which is equal to that of the state space constructed under basic state transition; if each vehicle within the control area does not conflict with any vehicle on the opposite lane (e.g., all vehicles within the control area are with straight movements.), the number of transitions reaches to the upper bound, because each state has the largest number of predecessor states in such case. Thus, two theorems are proposed for the number of transitions.
\newtheorem{remark}{Remark}
\begin{theorem}
\emph{The lower bound on the number of transitions is $O(N^4)$, where $N$ denotes the total number of vehicles within control area.}
\end{theorem}
\begin{proof}
    See \emph{Appendix \ref{appendix B-A}}.
\end{proof}
\begin{theorem}
\emph{The upper bound on the number of transitions is $O(N^6)$, where $N$ denotes the total number of vehicles within control area.}
\end{theorem}
\begin{proof}
    See \emph{Appendix \ref{appendix B-B}}.
\end{proof}
\par By Theorem 2-3, both the lower and upper bounds on the number of state transitions are polynomial with the number of vehicles within the control area.
\subsection{Analysis of Computational Time Complexity}
\label{sec5-C}
\par The theorem is proposed for computational complexity.
\begin{theorem}
\emph{The proposed strategy has the polynomial-time complexity of computation, which has a lower bound $O(N^4)$ and an upper bound $O(N^6)$, where $N$ denotes the total number of vehicles within control area.}
\end{theorem}
\begin{proof}
    See \emph{Appendix \ref{appendix C}}.
\end{proof}
\par By Theorem 4, the proposed optimal cooperative driving strategy can reduce the time complexity of computation from exponential to a small-degree polynomial.
\begin{remark}
In Appendix \ref{appendix C}, we get that each state transition will produce one computation with constant time. Then, it can be concluded that the time complexity of the proposed strategy is close to the lower bound in most cases, since there are many conflicts between vehicles in most cases and it is hard to satisfy the condition that each vehicle within the control area does not conflict with any vehicle on the opposite lane.
\end{remark}
\begin{remark}
There are a limited number of vehicles around an intersection so that the on-line computation can be realized as long as the time complexity of computation is a small-degree polynomial with the number of vehicles, and we do not need to pay much attention to whether the degree of a polynomial is 4 or 6.
\end{remark}
\begin{table}[h]
\renewcommand{\arraystretch}{1.3}
\centering
\caption{Parameters Setting in the Simulations.}
\centering
\begin{tabular}{c|c}
\Xhline{1.2pt}
{\bfseries Parameters} & {\bfseries Simulation setting}\\
\Xhline{1.0pt}
$a_\text{\text{max}}$,$a_{\text{min}}$ & $3 m/s^{2}$, $-5 m/s^{2}$\\
$v_{\text{max}}$,$v_{\text{min}}$ & $15m/s$, $0m/s$\\
$\qquad$$\qquad$ $l_c$ $\qquad$$\qquad$ & $\qquad$$\qquad$$\qquad$ $250m$ $\qquad$$\qquad$$\qquad$\\
$\Delta_{t,1}$,$\Delta_{t,2}$ & $1.5s$, $2s$\\
\Xhline{1.2pt}
\end{tabular}
\end{table}
\section{Simulation Results}
\label{sec6}
\subsection{Simulation Settings}
\label{sec6-A}
\par We design two kinds of simulations to verify the performance of the proposed optimal strategy. The first simulation aims to validate the global optimality of the proposed strategy. The second simulation aims to further demonstrate the performance of the proposed strategy in a continuous traffic process. In the simulations, we select the intersection shown in Fig. \ref{fig1} as the studied scenario. For the input lane of each direction, we assume that half of the vehicles will turn left and half will go straight.
\par As suggested in \cite{15,32}, the minimal safe gaps $\Delta_{t,1}$ and $\Delta_{t,2}$ to avoid collisions in conflict areas are set as $1.5s$ and $2s$, respectively. The length of control area $l_c$ is set as $250m$. In addition, the main parameter settings of the simulations are presented in TABLE II, and all simulations are carried out on the Visual Studio platform in a personal computer with an i7 CPU and a 16 GB RAM.
\par There are three kinds of performance indices evaluated in the simulations: \emph{\textcircled{1} total passing time:} it refers to the total time that all vehicles within the control area have passed the conflict area. In fact, the total passing time is exactly equal to the objective value of problem (\ref{equ6}) and can be calculated by (\ref{equ1b}). \emph{\textcircled{2} traffic throughput:} it refers to the total number of vehicles that have passed the conflict area in a specified period of time, which is utilized to evaluate the traffic efficiency in a continuous traffic process. \emph{\textcircled{3} average computation time:} it refers to the average computation time that is taken in a one-time planning procedure to get the vehicle sequence, which is used to evaluate the computational efficiency of strategies.
\par There are two strategies selected as the comparison strategies in the simulations: \emph{\textcircled{1} first-in-first-out (FIFO) strategy:} similar to \cite{15,16,24}, the typical FIFO strategy is utilized in the simulations, where the vehicles move into the conflict area following the sequence they enter the control area. Obviously, the time complexity of FIFO strategy is $O(N)$, where $N$ denotes the total number of vehicles within the control area. \emph{\textcircled{2} enumeration-based strategy:} it refers to the strategy that can obtain the globally optimal solution through enumerating all possible solutions, which is utilized to evaluate the global optimality of the proposed strategy. Obviously, the time complexity of enumeration-based strategy is $O(N!)$.
\subsection{Global Optimality Analysis}
\label{sec6-B}
\par In this simulation, we design an intersection scenario with $\alpha$ vehicles $(\alpha \in [5,24])$ that randomly distributed in the control area to verify the global optimality of the proposed strategy. Three strategies are applied in this scenario, i.e., FIFO strategy, enumeration-based strategy and the proposed strategy. The \emph{total passing time} and \emph{average computation time} are utilized as the performance indices in this simulation. For each $\alpha$, we simulate 20 times to take the average total passing time and computation time as shown in Fig. \ref{fig7} and TABLE \ref{Table3}. It should be pointed out that the computation time of enumeration-based strategy becomes extremely large when $\alpha > 18$. Thus, enumeration-based strategy is only applied in the scenarios where $\alpha \leq 18$ in this simulation.
\par According to Fig. \ref{fig7}, we can find that the total passing time of the proposed strategy is exactly equal to that of the enumeration-based strategy, which can obtain the globally optimal solution. Compared with the FIFO strategy, the total passing time of the proposed strategy is significantly decreased especially when the traffic volume is large. Furthermore, according to TABLE \ref{Table3}, the average computation time of the proposed strategy is short enough and can realize on-line computation in cooperative driving problem. However, the enumeration-based strategy is computationally intractable for a real-time implementation especially when there are a large number of vehicles within the control area.
\par Consequently, the proposed strategy can obtain the globally optimal solution with a limited computation time.
\begin{figure}[t]
\centering
\includegraphics[width=3in]{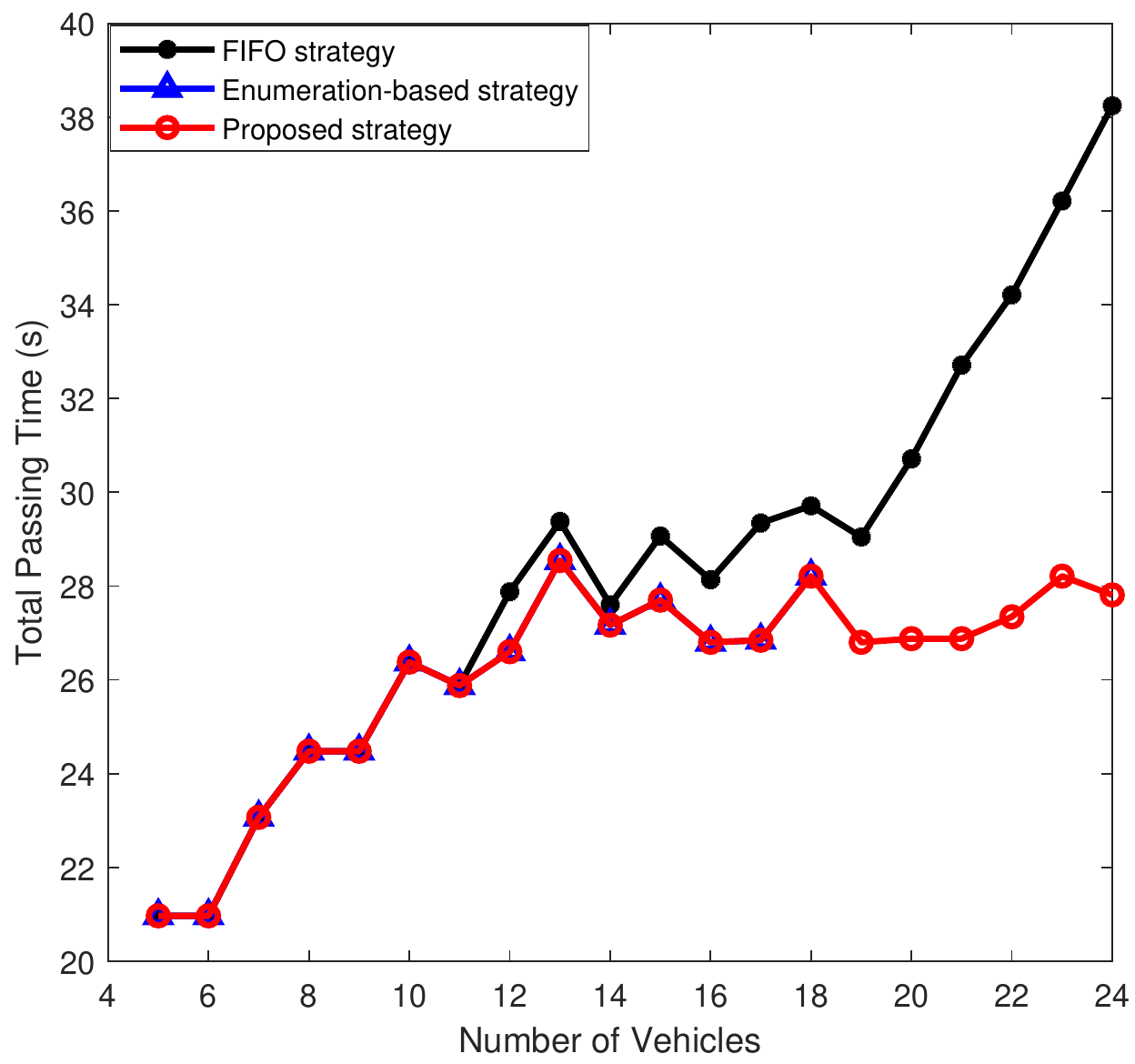}
\caption{The total passing time of different strategies with respect to the number of vehicles.}
\label{fig7}
\end{figure}
\begin{table}[t]
\renewcommand{\arraystretch}{1}
\centering
\caption{Average Computation Time of Different Strategies.}
\centering
\begin{tabular}{ccc}
\Xhline{1.2pt}
{\bfseries Number of} & \multirow{2}*{\bfseries Strategies} &\bfseries Average Computation\\
 {\bfseries Vehicles}&~&{\bfseries Time} $(ms)$\\
\Xhline{1.0pt}
\multirow{3}*{10} & {FIFO strategy} & {1}\\
~ & {Enumeration-based strategy} & {18} \\
~ & {Proposed strategy} & {3}\\
\hline
\multirow{3}*{14} & {FIFO strategy} & {1}\\
~ & {Enumeration-based strategy} & {3007} \\
~ & {Proposed strategy} & {6}\\
\hline
\multirow{3}*{18} & {FIFO strategy} & {1}\\
~ & {Enumeration-based strategy} & {1118876} \\
~ & {Proposed strategy} & {13}\\
\hline
\multirow{3}*{24} & {FIFO strategy} & {1}\\
~ & {Enumeration-based strategy} & {/} \\
~ & {Proposed strategy} & {21}\\
\Xhline{1.2pt}
\label{Table3}
\end{tabular}
\end{table}
\subsection{Comparison Results in Continuous Traffic Process}
\label{sec6-C}
\par In this simulation, the performance of the proposed strategy is further validated in a continuous traffic process. Similar to \cite{15,16}, we assume that the vehicles arrive in a Poisson Process at each input lane, and the average arrival rate of vehicles is denoted by $\lambda$ $veh/(lane \cdot h)$. In the continuous traffic process, the planning procedure is triggered when a new vehicle moves into the control area to reschedule the sequence of vehicles within the control area. In addition, similar to \cite{16,32}, we select the method introduced in \cite{31} to obtain the trajectory according to the assigned arrival time of vehicles.
\par Note that it is an extremely time-consuming process to obtain the globally optimal solution using the existing optimal strategies. Thus, the enumeration-based strategy is not adopted in this simulation and we select the FIFO strategy for comparisons. In addition, the \emph{traffic throughput} and \emph{average computation time} are selected as the performance indices in this simulation. For each parameter setting, we simulate a 10 minutes traffic process.
\begin{table}[t]
\renewcommand{\arraystretch}{1}
\centering
\caption{Comparison Results of Different Strategies.}
\centering
\begin{tabular}{cccc}
\Xhline{1.2pt}
{\bfseries Arrival Rate} & \multirow{2}*{\bfseries Strategies} & {\bfseries Traffic} & {\bfseries Computation}\\
 {$veh/(lane \cdot h)$} & ~ & {\bfseries Throughput} & {\bfseries Time} $(ms)$\\
\Xhline{1pt}
\multirow{2}*{400} & {FIFO strategy} & {230} & {1}\\
~ & {Proposed strategy} & {231} & {1}\\
\hline
\multirow{2}*{450} & {FIFO strategy} & {260} & {1}\\
~ & {Proposed strategy} & {261} & {2}\\
\hline
\multirow{2}*{500} & {FIFO strategy} & {259} & {1}\\
~ & {Proposed strategy} & {328} & {1}\\
\hline
\multirow{2}*{550} & {FIFO strategy} & {263} & {1}\\
~ & {Proposed strategy} & {353} & {8}\\
\hline
\multirow{2}*{600} & {FIFO strategy} & {258} & {1}\\
~ & {Proposed strategy} & {382} & {13}\\
\Xhline{1.2pt}
\label{Table4}
\end{tabular}
\end{table}
\begin{figure}[t]
\centering
\includegraphics[width=3in]{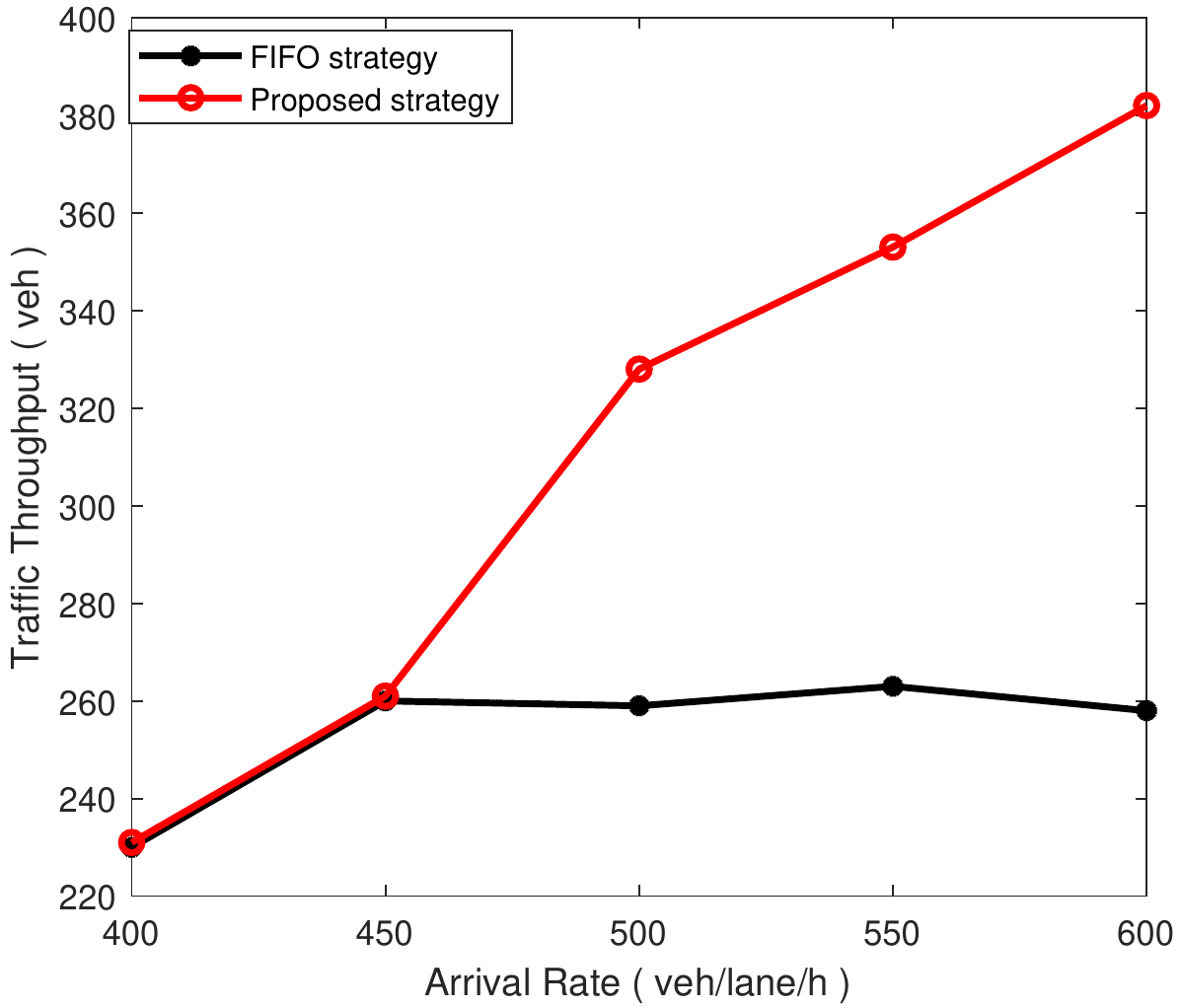}
\caption{Throughput of different strategies with respect to the arrival rate.}
\label{fig8}
\end{figure}
\par According to TABLE \ref{Table4} and Fig. \ref{fig8}, it is obvious that the proposed strategy outperforms FIFO strategy in all simulations in terms of traffic throughput, and the difference in traffic throughput of these two strategies increases with the arrival rate. Furthermore,  the traffic throughput of FIFO strategy reaches a saturation state as the arrival rate increases, since the traffic becomes seriously congested
and most vehicles block in the upstream when FIFO strategy is applied in the scenario and the arrival rate of vehicles is large. More significantly, the average computation time of the proposed strategy is short enough in all simulations, which can entirely guarantee on-line computation when implementing optimal cooperative driving.
\section{Conclusion}
\label{sec7}
\par In this paper, we study the problem of cooperative driving at signal-free intersections, aiming to obtain the globally optimal solution within an efficient computation time. Taking advantage of the conflict relations between vehicles, we show that the large-scale planning problem can be reformulated as a small-size multi-stage decision problem using the idea of dynamic programming, which can reduce the time complexity from exponential to a small-degree polynomial. In the future, the proposed strategy can be extended to other traffic scenarios (e.g., multi-intersection road networks) and even for the single machine total tardiness problem in operational research field.
\appendices
\section{Proof of Theorem 1}
\label{appendix A}
\begin{proof}
Recalling that the number of vehicles on lane 1, 2, 3 and 4 is denoted by $\hat n_1$, $\hat n_2$, $\hat n_3$ and $\hat n_4$ respectively, and assuming that $\hat n_1>0$, $\hat n_2>0$, $\hat n_3>0$ and $\hat n_4>0$. To simplify the descriptions, we give a symbolic function, i.e.,
\begin{equation}
\label{equ20}
\Gamma(x)=
\begin{cases}
0& x=0\\
1& x \neq 0
\end{cases}.
\nonumber
\end{equation}
 Then, the number of states in the solution space can be analyzed by categories. As for $\Gamma(n_1)+\Gamma(n_2)+\Gamma(n_3)+\Gamma(n_4)=0$, there is only one state satisfying condition, i.e., the initial state in stage 0. As for $\Gamma(n_1)+\Gamma(n_2)+\Gamma(n_3)+\Gamma(n_4)=1$, there are $(\hat n_1+\hat n_2+\hat n_3+\hat n_4)$ states satisfying condition, e.g., $s_{1}\begin{pmatrix}\begin{smallmatrix}1 & 0 \\ 0 & 0\end{smallmatrix}\end{pmatrix}$,
$s_{1}\begin{pmatrix}\begin{smallmatrix}2 & 0 \\ 0 & 0\end{smallmatrix}\end{pmatrix}$,
$s_{2}\begin{pmatrix}\begin{smallmatrix}0 & 1 \\ 0 & 0\end{smallmatrix}\end{pmatrix}$. As for $\Gamma(n_1)+\Gamma(n_2)+\Gamma(n_3)+\Gamma(n_4)=2$, there are $(2\hat n_1\hat n_2+2\hat n_1\hat n_3+2\hat n_1 \hat n_4+2\hat n_2\hat n_3+2\hat n_2\hat n_4+2\hat n_3\hat n_4)$ states satisfying condition, e.g., $s_{1}\begin{pmatrix}\begin{smallmatrix}1 & 1 \\ 0 & 0\end{smallmatrix}\end{pmatrix}$,
$s_{1}\begin{pmatrix}\begin{smallmatrix}1 & 0 \\ 1 & 0\end{smallmatrix}\end{pmatrix}$,
$s_{2}\begin{pmatrix}\begin{smallmatrix}0 & 1 \\ 1 & 0\end{smallmatrix}\end{pmatrix}$. As for $\Gamma(n_1)+\Gamma(n_2)+\Gamma(n_3)+\Gamma(n_4)=3$, there are $(3\hat n_1\hat n_2\hat n_3+3\hat n_1\hat n_2\hat n_4+3\hat n_1\hat n_3\hat n_4+3\hat n_2\hat n_3\hat n_4)$ states satisfying condition, e.g., $s_{1}\begin{pmatrix}\begin{smallmatrix}1 & 1 \\ 1 & 0\end{smallmatrix}\end{pmatrix}$,
$s_{1}\begin{pmatrix}\begin{smallmatrix}1 & 1 \\ 0 & 1\end{smallmatrix}\end{pmatrix}$,
$s_{2}\begin{pmatrix}\begin{smallmatrix}0 & 1 \\ 1 & 2\end{smallmatrix}\end{pmatrix}$. As for $\Gamma(n_1)+\Gamma(n_2)+\Gamma(n_3)+\Gamma(n_4)=4$, there are $(4\hat n_1\hat n_2\hat n_3\hat n_4)$ states satisfying condition, e.g., $s_{1}\begin{pmatrix}\begin{smallmatrix}1 & 1 \\ 1 & 1\end{smallmatrix}\end{pmatrix}$,
$s_{1}\begin{pmatrix}\begin{smallmatrix}1 & 1 \\ 2 & 1\end{smallmatrix}\end{pmatrix}$,
$s_{2}\begin{pmatrix}\begin{smallmatrix}3 & 1 \\ 1 & 2\end{smallmatrix}\end{pmatrix}$. Based on the above descriptions, we can obtain the total number of states $N_s$ in the state space as
\begin{equation}
\begin{split}
 N_{s}=4\hat n_1 \hat n_2 \hat n_3 \hat n_4 + 3(\hat n_1 \hat n_2 \hat n_3+\hat n_1 \hat n_2 \hat n_4 +\hat n_1 \hat n_3 \hat n_4\\ + \hat n_2 \hat n_3 \hat n_4)+ 2(\hat n_1 \hat n_2 + \hat n_1 \hat n_3+\hat n_1 \hat n_4+\hat n_2 \hat n_3\\ +\hat n_2 \hat n_4 + \hat n_3 \hat n_4)+ \hat n_1 + \hat n_2 + \hat n_3 +\hat n_4 + 1.
 \label{equ26}
 \end{split}
\nonumber
\end{equation}
\par Assuming that the total number of vehicles in control
area is $N(N \geq 4)$ and $\hat n_1=\hat n_2=\hat n_3=\hat n_4=\frac {N}{4}$, we have
\begin{equation}
 N_{s}=\frac{N^4}{64}+\frac{3N^3}{16}+\frac{3N^2}{4}+N+1.
\nonumber
\end{equation}
It is obvious that the number of states is $O(N^4)$.
\end{proof}
\vspace{-0.3cm}
\section{Proof of Theorem 2-3}
\par Recalling that the number of vehicles on lane 1, 2, 3 and 4 is denoted by $\hat n_1$, $\hat n_2$, $\hat n_3$ and $\hat n_4$ respectively, and assuming that $\hat n_1>0$, $\hat n_2>0$, $\hat n_3>0$ and $\hat n_4>0$.
\label{appendix B}
\vspace{-0.5cm}
\subsection{Proof of Theorem 2}
\label{appendix B-A}
\begin{proof} The number of transitions reaches the lower bound $N_t^l$ under basic state transition. Usually, each state has four predecessor states except the initial state, i.e.,
\begin{equation}
N_{t,\text{max}}^l = 4(N_{s}-1),
\label{equ27}
\nonumber
\end{equation}
where $N_{t,\text{max}}^l$ denotes the possible maximum of transitions under basic state transition. However, there exist several infeasible transitions, which need to be subtracted from $N_{t,\text{max}}^l$. As for such states as $s_{1}\begin{pmatrix}\begin{smallmatrix}n_1 & n_2 \\ n_3 & n_4\end{smallmatrix}\end{pmatrix}$, there are generally four kinds of predecessor states, i.e., $s_{1}\begin{pmatrix}\begin{smallmatrix}n_1-1 & n_2 \\ n_3 & n_4\end{smallmatrix}\end{pmatrix}$, $s_{2}\begin{pmatrix}\begin{smallmatrix}n_1-1 & n_2 \\ n_3 & n_4\end{smallmatrix}\end{pmatrix}$, $s_{3}\begin{pmatrix}\begin{smallmatrix}n_1-1 & n_2 \\ n_3 & n_4\end{smallmatrix}\end{pmatrix}$, and $s_{4}\begin{pmatrix}\begin{smallmatrix}n_1-1 & n_2 \\ n_3 & n_4\end{smallmatrix}\end{pmatrix}$. We now discuss whether these predecessor states exist in state space. \textcircled{1} if $n_1=1$, state $s_{1}\begin{pmatrix}\begin{smallmatrix}n_1-1 & n_2 \\ n_3 & n_4\end{smallmatrix}\end{pmatrix}$ is infeasible, where $n_2 \in [0,\hat n_2]$, $n_3 \in [0,\hat n_3]$ and $n_4 \in [0,\hat n_4]$. Thus, there are $[(\hat n_2+1)(\hat n_3+1)(\hat n_4+1)]$ states belonging to this category. \textcircled{2} if $n_2=0$, $s_{2}\begin{pmatrix}\begin{smallmatrix}n_1-1 & n_2 \\ n_3 & n_4\end{smallmatrix}\end{pmatrix}$ is infeasible, where $n_1 \in [1,\hat n_1]$, $n_3 \in [0,\hat n_3]$ and $n_4 \in [0,\hat n_4]$. Thus, there are $[\hat n_1(\hat n_3+1)(\hat n_4+1)]$ states belonging to this category. \textcircled{3} if $n_3=0$, $s_{3}\begin{pmatrix}\begin{smallmatrix}n_1-1 & n_2 \\ n_3 & n_4\end{smallmatrix}\end{pmatrix}$ is infeasible, where $n_1 \in [1,\hat n_1]$, $n_2 \in [0,\hat n_2]$ and $n_4 \in [0,\hat n_4]$. Thus, there are $[\hat n_1(\hat n_2+1)(\hat n_4+1)]$ states belonging to this category. \textcircled{4} if $n_4=0$, $s_{4}\begin{pmatrix}\begin{smallmatrix}n_1-1 & n_2 \\ n_3 & n_4\end{smallmatrix}\end{pmatrix}$ is infeasible, where $n_1 \in [1,\hat n_1]$, $n_2 \in [0,\hat n_2]$ and $n_3 \in [0,\hat n_3]$. Thus, there are $[\hat n_1(\hat n_2+1)(\hat n_3+1)]$ states of this category. In summary, as for such states as $s_{1}\begin{pmatrix}\begin{smallmatrix}n_1 & n_2 \\ n_3 & n_4\end{smallmatrix}\end{pmatrix}$, the number of infeasible transitions $N_{t,inf,1}$ that need to be subtracted from $N_{t,\text{max}}^l$ is
\begin{equation}
\begin{split}
N_{t,inf,1} = (\hat n_2+1)(\hat n_3+1)(\hat n_4+1) + \hat n_1(\hat n_3+1)(\hat n_4+1)\\ + \hat n_1(\hat n_2+1)(\hat n_4+1) + \hat n_1(\hat n_2+1)(\hat n_3+1)
 \label{equ28}
 \end{split}
 \nonumber
\end{equation}
\vspace{-0.5cm}
\par\noindent Similarly, as for such states as $s_{2}\begin{pmatrix}\begin{smallmatrix}n_1 & n_2 \\ n_3 & n_4\end{smallmatrix}\end{pmatrix}$, $s_{3}\begin{pmatrix}\begin{smallmatrix}n_1 & n_2 \\ n_3 & n_4\end{smallmatrix}\end{pmatrix}$, and $s_{4}\begin{pmatrix}\begin{smallmatrix}n_1 & n_2 \\ n_3 & n_4\end{smallmatrix}\end{pmatrix}$, the number of transitions that need to be subtracted from $N_{t,\text{max}}^l$ is
\begin{equation}
\begin{split}
N_{t,inf,2} = (\hat n_1+1)(\hat n_3+1)(\hat n_4+1) + \hat n_2(\hat n_3+1)(\hat n_4+1)\\ + \hat n_2(\hat n_1+1)(\hat n_4+1) + \hat n_2(\hat n_1+1)(\hat n_3+1),
 \label{equ29}
 \end{split}
 \nonumber
\end{equation}
\vspace{-0.5cm}
\begin{equation}
\begin{split}
N_{t,inf,3} = (\hat n_1+1)(\hat n_2+1)(\hat n_4+1) + \hat n_3(\hat n_2+1)(\hat n_4+1)\\ + \hat n_3(\hat n_1+1)(\hat n_4+1) + \hat n_3(\hat n_1+1)(\hat n_2+1),
 \label{equ30}
 \end{split}
 \nonumber
\end{equation}
\vspace{-0.5cm}
\begin{equation}
\begin{split}
N_{t,inf,4} = (\hat n_1+1)(\hat n_2+1)(\hat n_3+1) + \hat n_4(\hat n_2+1)(\hat n_3+1)\\ + \hat n_4(\hat n_1+1)(\hat n_3+1) + \hat n_4(\hat n_1+1)(\hat n_2+1).
 \label{equ31}
 \end{split}
 \nonumber
\end{equation}
\vspace{-0.5cm}
\par\noindent In addition, there are four states that directly connect to the initial state, i.e., $s_{1}\begin{pmatrix}\begin{smallmatrix}1 & 0 \\ 0 & 0\end{smallmatrix}\end{pmatrix}$, $s_{2}\begin{pmatrix}\begin{smallmatrix}0 & 1 \\ 0 & 0\end{smallmatrix}\end{pmatrix}$, $s_{3}\begin{pmatrix}\begin{smallmatrix}0 & 0 \\ 0 & 1\end{smallmatrix}\end{pmatrix}$ and $s_{4}\begin{pmatrix}\begin{smallmatrix}0 & 0 \\ 0 & 1\end{smallmatrix}\end{pmatrix}$, which produces four special transitions. Therefore, the lower bound is summarized as
\begin{equation}
\label{equ32}
\begin{split}
N_{t}^l = N_{t,\text{max}}^l - N_{t,inf,1} - N_{t,inf,2} - N_{t,inf,3} - N_{t,inf,4}+4\\=16\hat n_1 \hat n_2 \hat n_3 \hat n_4 + 8(\hat n_1 \hat n_2 \hat n_4+\hat n_1 \hat n_2 \hat n_4 + \hat n_1 \hat n_3 \hat n_4 + \hat n_2 \hat n_3 \hat n_4)\\ + 2(\hat n_1 \hat n_2 + \hat n_1 \hat n_3+ \hat n_1 \hat n_4 + \hat n_2 \hat n_3 + \hat n_2 \hat n_4 + \hat n_3 \hat n_4)\\ -2(\hat n_1 + \hat n_2 + \hat n_3 + \hat n_4)+4.
\end{split}
\nonumber
\end{equation}
\vspace{-0.5cm}
\par Assuming that the total number of vehicles in control
area is $N(N \geq 4)$ and $\hat n_1=\hat n_2=\hat n_3=\hat n_4=\frac {N}{4}$, we have
\begin{equation}
 N_{t}^{l}=\frac{N^4}{16}+\frac{N^3}{2}+\frac{3N^2}{4}-2N+4.
 \nonumber
\end{equation}
Thus, the lower bound on number of transitions is $O(N^4)$.
\end{proof}
\vspace{-0.5cm}
\subsection{Proof of Theorem 3}
\label{appendix B-B}
\begin{proof} According to Property 2 in \emph{Section \Rmnum{3}-C}, we know that removing one state transition may reproduce multiple transitions which do not connect the states in adjacent stages. Thus, the total number of transitions is
\begin{equation}
    N_{t} = N_{t}^l - N_{t,\text{rem}} + N_{t,\text{rep}},
\label{equ33}
\end{equation}
where $N_{t,\text{rem}}$ denotes the number of removed transitions, and $N_{t,\text{rep}}$ denotes the number of reproduced transitions.
\par According to Property 1 in \emph{Section \Rmnum{3}-C}, the reproduced transitions represent a group of vehicles which do not conflict with each other. Firstly, we aim to find the number of the reproduced transitions that connect to a particular state
$s_{r}\begin{pmatrix}\begin{smallmatrix}n_1^0+\delta_{n_1} & n_2 \\ n_3^0+\delta_{n_3} & n_4\end{smallmatrix}\end{pmatrix}$, where we assume that the $n_1^{th} (n_1^{th}\in [n_1^0+1,n_1^0+\delta_{n_1}])$ vehicle on lane 1 has no conflict with the $n_3^{th} (n_3^{th}\in [n_3^0+1,n_3^0+\delta_{n_3}])$ on lane 3. Thus, there are $\delta_{n_1}\delta_{n_3}$ reproduced transitions connecting to state $s_{r}\begin{pmatrix}\begin{smallmatrix}n_1^0+\delta_{n_1} & n_2 \\ n_3^0+\delta_{n_3} & n_4\end{smallmatrix}\end{pmatrix}$ when $n_2$, $n_4$ and $r$ are fixed. Secondly, according to \emph{Section \Rmnum{4}-B}, the total number of transitions reaches the upper bound $N_t^u$ when each vehicle within the control area does not conflict with any vehicle on the opposite lane. In such case, we aim to find the total number of reproduced transitions for state $s_{r}\begin{pmatrix}\begin{smallmatrix}n_1^0+\delta_{n_1} & n_2 \\ n_3^0+\delta_{n_3} & n_4\end{smallmatrix}\end{pmatrix}$, where $\delta_{n_1}$ and $\delta_{n_3}$ are set as different values. Obviously, we have $\delta_{n_1} \in [1,\hat n_1]$ and $\delta_{n_3} \in [1,\hat n_3]$. Thus, the total number of the reproduced transitions connecting to this kind of states where $n_2$, $n_4$ and $r$ are fixed is
\begin{equation}
    \sum_{\delta_{n_1}=1}^{\hat n_1} \sum_{\delta_{n_3}=1}^{\hat n_3} \delta_{n_1}\delta_{n_3}=\frac{\hat n_1(\hat n_1+1) \hat n_3(\hat n_3+1)}{4}.
\label{equ34}
\nonumber
\end{equation}
Then, we consider that $n_2$, $n_4$ and $r$ can set as different values. As for $n_2 \in [1,\hat n_2]$ and $n_4 \in [1,\hat n_4]$, $r \in \{2,4\}$. As for $n_2=0$ and $n_4 \in [1,\hat n_4]$, $r=4$. As for $n_4=0$ and $n_2 \in [1,\hat n_2]$, $r=2$. Therefore, for such states as $s_{r}\begin{pmatrix}\begin{smallmatrix}n_1^0+\delta_{n_1} & n_2 \\ n_3^0+\delta_{n_3} & n_4\end{smallmatrix}\end{pmatrix}$, the total number of reproduced transitions $N_{t,\text{rep},1}$ is
\begin{equation}
N_{t,\text{rep},1} = (2\hat n_2\hat n_4+\hat n_2+\hat n_4)\frac{\hat n_1(\hat n_1+1)\hat n_3(\hat n_3+1)}{4}.
\label{equ35}
\nonumber
\end{equation}
Similarly, for such states as $s_{r}\begin{pmatrix}\begin{smallmatrix}n_1 & n_2^0+\delta_{n_2} \\ n_3 & n_4^0+\delta_{n_4}\end{smallmatrix}\end{pmatrix}$, the total number of reproduced transitions is
\begin{equation}
N_{t,\text{rep},2} = (2\hat n_1\hat n_3+\hat n_1+\hat n_3)\frac{\hat n_2(\hat n_2+1)\hat n_4(\hat n_4+1)}{4}.
\label{equ36}
\nonumber
\end{equation}
In addition, there are $(2\hat n_1\hat n_3+2\hat n_2\hat n_4)$ reproduced transitions from the initial state to such states as $s_{r}\begin{pmatrix}\begin{smallmatrix}n_1 & 0 \\ n_3 & 0\end{smallmatrix}\end{pmatrix}$ and $s_{r}\begin{pmatrix}\begin{smallmatrix}0 & n_2 \\ 0 & n_4\end{smallmatrix}\end{pmatrix}$. Based on above descriptions, we have
\begin{equation}
    N_{t,\text{rep}} =  N_{t,\text{rep},1}+ N_{t,\text{rep},2}+2\hat n_1\hat n_3+2\hat n_2\hat n_4.
\label{equ37}
\end{equation}
Furthermore, it is easy to obtain $N_{t,\text{rem}}$, i.e.,
\begin{equation}
    N_{t,\text{rem}} = 2\hat n_1\hat n_3(\hat n_2+1)(\hat n_4+1)+2\hat n_2\hat n_4(\hat n_1+1)(\hat n_3+1).
\label{equ38}
\end{equation}
Finally, according to (\ref{equ33}), (\ref{equ37}) and (\ref{equ38}), we can obtain the upper bound on the number of transitions, i.e.,
\begin{equation}
\begin{split}
N_{t}^u = N_{t}^l + (2\hat n_2\hat n_4+\hat n_2+\hat n_4)\frac{\hat n_1(\hat n_1+1)\hat n_3(\hat n_3+1)}{4}\\+(2\hat n_1\hat n_3+\hat n_1+\hat n_3)\frac{\hat n_2(\hat n_2+1)\hat n_4(\hat n_4+1)}{4}+2\hat n_1\hat n_3+2\hat n_2\hat n_4\\-2\hat n_1\hat n_3(\hat n_2+1)(\hat n_4+1)-2\hat n_2\hat n_4(\hat n_1+1)(\hat n_3+1).
\end{split}
\nonumber
\end{equation}
\vspace{-0.5cm}
\par Assuming that the total number of vehicles in control
area is $N(N \geq 4)$ and $\hat n_1=\hat n_2=\hat n_3=\hat n_4=\frac {N}{4}$, we have
\begin{equation}
 N_{t}^{u}=\frac{N^6}{4096}+\frac{3N^5}{1024}+\frac{15N^4}{256}+\frac{25N^3}{64}+\frac{3N^2}{4}-2N.
 \nonumber
\end{equation}
Thus, the upper bound on number of transitions is $O(N^6)$.
\end{proof}
\vspace{-0.3cm}
\section{Proof of Theorem 4}
\label{appendix C}
\begin{proof}
    In Algorithm 1, for each state transition, we need to calculate and update the objective value of current state. To further reduce the complexity, instead of using Eq. (\ref{equ11a})-Eq. (\ref{equ12'}), it is easy to obtain the objective value according to the conflict relations between vehicles. Specifically, if multiple vehicles can obtain the right of way during current state transition and these vehicles are divided into two continuous queues in different lanes (e.g., Fig. \ref{fig4-new}), we just need to calculate the arrival time assigned to the last vehicle in each queue and then pick the maximum as the objective value, i.e.,
\begin{subequations}
\label{equ13''}
\begin{equation}
t_{\text{assign},r,l} = \max\left(t_{\text{assign},r,f}+(N_r-1)\Delta_{t,1},t_{\min,r,l}\right),
\nonumber
\end{equation}
\begin{equation}
t_{\text{assign},o,l} = \max\left(t_{\text{assign},o,f}+(N_o-1)\Delta_{t,1},t_{\min,o,l}\right),
\nonumber
\end{equation}
\begin{equation}
J =\max\left(t_{\text{assign},r,l},t_{\text{assign},o,l}\right),
\nonumber
\end{equation}
\end{subequations}
where lane $o$ is the opposite lane of lane $r$. $t_{\text{assign},r,l}$ and $t_{\text{assign},o,l}$ denote the arrival time assigned to the last vehicle of the queue in lane $r$ and lane $o$ respectively. $t_{\text{assign},r,f}$ and $t_{\text{assign},o,f}$ denote the arrival time assigned to the first vehicle of the queue in lane $r$ and lane $o$ respectively. $N_r$ and $N_o$ denote the size of the queue in lane $r$ and lane$o$ respectively. $t_{\min,r,l}$ and $t_{\min,o,l}$ denote the minimal arrival time of the last vehicle of the queue in lane $r$ and lane $o$ respectively.
\par Obviously, the time complexity of calculating the objective value is a constant time for each transition, since the computation time does not change with the size of the algorithm input (i.e., the total number of vehicles). Therefore, based on Theorem 2-3, it is obvious that the optimal cooperative driving strategy has the polynomial-time complexity of computation, which has a lower bound $O(N^4)$ and an upper bound $O(N^6)$, where $N$ denotes the number of vehicles.
\end{proof}
\bibliographystyle{IEEEtran}
\bibliography{IEEEabrv,IEEEexample}
\end{document}